\documentclass[11pt]{amsart}
\usepackage{amsfonts}
\usepackage{amsmath}
\usepackage{amssymb}
\usepackage{color}
\usepackage{graphicx}
\usepackage[all]{xy}   
\usepackage{xspace}
\usepackage{fullpage}
 \font \eightrm=cmr8

\addtocounter{MaxMatrixCols}{2} %% allows for a 12 x 12 pmatrix

 \newcommand{\nc}{\newcommand}

\newtheorem{thm}{Theorem}

\newtheorem{cor}[thm]{Corollary}

\newtheorem{lem}[thm]{Lemma}
\newtheorem{prop}[thm]{Proposition}
\newtheorem{defn}{Definition}

\DeclareMathOperator{\Lin}{Lin}    %% convolution algebra
\nc{\ignore}[1]{{}}
\nc{\mrm}[1]{{\rm #1}}
\nc{\dirlim}{\displaystyle{\lim_{\longrightarrow}}\,}
\nc{\invlim}{\displaystyle{\lim_{\longleftarrow}}\,}
\nc{\vep}{\varepsilon} \nc{\ep}{\epsilon}
\nc{\sigmat}{\widetilde\sigma}
\nc{\ostar}{\overline{*}}

\nc{\mchar}{\mrm{Char}}
\nc{\Hom}{\mrm{Hom}}
\nc{\id}{\mrm{id}}

\nc{\remark}{\noindent{\bf{Remark:}}}
\nc{\remarks}{\noindent{\bf{Remarks:}}}

 \nc{\delete}[1]{}
 \nc{\grad}[1]{^{({#1})}}
 \nc{\fil}[1]{_{#1}}

\nc{\BA}{{\Bbb A}} \nc{\CC}{{\Bbb C}} \nc{\DD}{{\Bbb D}}
\nc{\EE}{{\Bbb E}} \nc{\FF}{{\Bbb F}} \nc{\GG}{{\Bbb G}}
\nc{\HH}{{\Bbb H}} \nc{\LL}{{\Bbb L}} \nc{\NN}{{\Bbb N}}
\nc{\PP}{{\Bbb P}} \nc{\QQ}{{\Bbb Q}} \nc{\RR}{{\Bbb R}}
\nc{\TT}{{\Bbb T}} \nc{\VV}{{\Bbb V}} \nc{\ZZ}{{\Bbb Z}}
\nc{\Cal}[1]{{\mathcal {#1}}}
\nc{\mop}[1]{\mathop{\hbox {\rm #1} }}
\nc{\smop}[1]{\mathop{\hbox {\eightrm #1} }}
\nc{\mopl}[1]{\mathop{\hbox {\rm #1} }\limits}
\nc{\frakg}{{\frak g}}
\nc{\g}[1]{{\frak {#1}}}
\def \restr#1{\mathstrut_{\textstyle |}\raise-8pt\hbox{$\scriptstyle #1$}}
\def \srestr#1{\mathstrut_{\scriptstyle |}\hbox to
  -1.5pt{}\raise-4pt\hbox{$\scriptscriptstyle #1$}}
\nc{\wt}{\widetilde}
\nc{\wh}{\widehat}
\nc{\un}{\hbox{\bf 1}}
\nc{\redtext}[1]{\textcolor{red}{\tt #1}}
\nc{\bluetext}[1]{\textcolor{blue}{#1}}
\nc{\comment}[1]{[[{\tt {#1}}]] }
\nc{\R}{{\mathbb R}}

\nc\fleche[1]{\mathop{\hbox to #1 mm{\rightarrowfill}}\limits}
\def\semi{\mathrel{\times}\kern -.85pt\joinrel\mathrel{\raise 1.4pt\hbox{${\scriptscriptstyle |}$}}}

\author{Kurusch Ebrahimi-Fard}
\address{Departamento de F\'{i}sica Te\'orica, Universidad de Zaragoza, 50009 Zaragoza, Spain. 
	       {On leave from Univ.~de Haute Alsace, Mulhouse, France}}
    \email{kef@unizar.es, kurusch.ebrahimi-fard@uha.fr}         
 \urladdr{http://www.th.physik.uni-bonn.de/th/People/fard/}

\author{Fr\'ed\'eric Patras}
\address{Universit\'e de Nice,
		Laboratoire J.-A.~Dieudonn\'e
	       	UMR 6621, CNRS,
               	Parc Valrose,
               	06108 Nice Cedex 02, France}
    \email{patras@unice.fr}
 \urladdr{www-math.unice.fr/~patras}

\begin{document}
%%%%%%%%%%%%%%%%%%%%%%%%%%%%%%%%%%%%%%%%%%%%%
%%%%%%%%%%%%%%%%%%%%%%%%%%%%%%%%%%%%%%%%%%%%%

\title{Exponential renormalization}
	
%%%%%%%%%%%%%%%%%%%%%%%%%%%%%%%%%%%%% %%%%%%%%%%%%%%%%%%%%%%%%
\date{June 12th, 2010}
%\noindent {\footnotesize{${}\phantom{a}$ Mathematics Subject
%Classification 2000}: 12H20,37C10 }}
%%%%%%%%%%%%%%%%%%%%%%%%%%%%%%%%%%%%%%%%%%%%%%%%%%%%%%%%%%%%%

\begin{abstract}
Moving beyond the classical additive and multiplicative approaches, we present an ``exponential'' method for perturbative renormalization. Using Dyson's identity for Green's functions as well as the link between the Fa\`a di Bruno Hopf algebra and  the Hopf algebras of Feynman graphs, its relation to the composition of formal power series is analyzed. Eventually, we argue that the new method has several attractive features and encompasses the BPHZ method. The latter can be seen as a special case of the new procedure for renormalization scheme maps with the Rota--Baxter property. To our best knowledge, although very natural from group-theoretical and physical points of view, several ideas introduced in the present paper seem to be new (besides the exponential method, let us mention the notions of counterfactors and of order $n$ bare coupling constants).
\end{abstract}

\maketitle
%%%%%%%%%%%%%%%%%%%%%%%%%%%%%%%%%%%%%%%%%%%%%%%%%%%%%%%%%%%%%
\tableofcontents
%%%%%%%%%%%%%%%%%%%%%%%%%%%%%%%%%%%%%%%%%%%%%%%%%%%%%%%%%%%%%

\section{Introduction}

Renormalization theory \cite{CasKen,Collins,Del,IZ1980} plays a major role in the perturbative approach to quantum field theory (QFT). Since its inception in the late 1930s \cite{Brown} it has evolved from a highly technical and difficult set of tools, mainly used in precision calculations in high energy particle physics, into a fundamental physical principle encoded by the modern notion of the renormalization group.    

Recently, Alain Connes, Dirk Kreimer, Matilde Marcolli and collaborators developed a compelling mathematical setting capturing essential parts of the algebraic and combinatorial structure underlying the so-called BPHZ renormalization procedure in perturbative QFT \cite{CKI,CKII,CKIII,CM2008,kreimer2}. The essential notion appearing in this approach is the one of combinatorial Hopf algebras. The latter typically consists of a graded vector space where the homogeneous components are spanned by finite sets of combinatorial objects, such as planar or non-planar rooted trees, or Feynman graphs, and the Hopf algebraic structures are given by particular constructions on those objects. For a particular QFT the set of Feynman rules corresponds to a multiplicative map from such a combinatorial Hopf algebra, generated, say, by one-particle irreducible (1PI) ultraviolet (UV) superficially divergent diagrams, into a commutative unital target algebra. This target algebra essentially reflects the regularization scheme.

The process of renormalization in perturbative QFT can be performed in many different ways~\cite{Collins,IZ1980}. A convenient framework is provided by dimensional regularization (DR). It implies a target algebra of regularized probability amplitudes equipped with a natural Rota--Baxter (RB) algebra structure. The latter encodes nothing but minimal subtraction (MS). Introducing a combinatorial Hopf algebra of Feynman graphs in the context of $\phi^3$-theory (in 6 dimensions) allows for example to reformulate the BPHZ renormalization method for Feynman graphs, in terms of a Birkhoff--Wiener--Hopf (BWH) decomposition inside the group of dimensionally regularized characters \cite{kreimer2,CKII}. As it turns out, Bogoliubov's recursive renormalization process is then best encoded by Atkinson's recursion for noncommutative Rota--Baxter algebras, the solution of which was obtained in the form of a closed formula in \cite{EFMP}.  Following Kreimer \cite{kreimer1}, Walter van Suijlekom extended the Hopf algebra approach to perturbative renormalization of gauge theories \cite{vanSu1,vanSu2}. 

The Connes--Kreimer approach focussed originally on DR+MS but can actually be extended to other regularization schemes, provided the subtraction method corresponds to a Rota--Baxter algebra structure. It applies for example to zero momentum subtraction as shown in \cite{EFGP}. However, essential parts of this algebraic machinery are not anymore available once the RB property is lost. More precisely, the remarquable result that Bogoliubov's classical renormalization formulae give birth to Hopf algebra characters and are essentially equivalent to the BWH decomposition of Hopf algebra characters is lost if the renormalization scheme map is not RB \cite{CKI,CKII}. 

Two remarks are in order. First, more insights from an algebraic point of view are needed in this particular direction. As a contribution to the subject, we propose and study in the last section of the present paper a non-MS scheme within DR which is not of Rota--Baxter type. Second, the characterization of the BPHZ method in terms of BWH decomposition might be too restrictive, as it excludes possible subtraction schemes that do not fall into the class of Rota--Baxter type ones. In this paper we present an exponential algorithm to perform perturbative renormalization (the term ``exponential'' refers to the way the algorithm is constructed and was also chosen for its similarity with the classical ``additive'' and ``multiplicative'' terminologies). One advantage of this method, besides its group-theoretical naturality, is that it does not rely on the Rota--Baxter property. Indeed, the exponential method is less restrictive than the BPHZ method in the Hopf algebraic picture. It only requires a projector $P_-$ (used to isolate the divergences of regularized amplitudes) such that the image of the associated orthogonal projector, $P_+:=id-P_-$, forms a subalgebra. This constraint on the image of $P_+$ reflects the natural assumption that products of finite regularized amplitudes are supposed to be finite. 

Let us mention that the very process of exponential renormalization leads to the introduction of new objects and ideas in the algebro-combinatorial approach to perturbative QFT. Particularly promissing are the ones of counterfactors and order $n$ bare coupling constants, that fit particularly well some widespread ideas that do not always come with a rigorous mathematical foundation such as the one that ``in the end everything boils down in perturbative QFT to power series substitutions''. The notion of  order $n$ bare coupling constants makes such a statement very precise from the algebraic point of view.

Let us also mention that the exponential method is a further development of ideas sketched in our earlier paper~\cite{KEFPatras} that pointed at a natural link between renormalization techniques and fine properties of Lie idempotents, with a particular emphasis on the family of Zassenhaus Lie idempotents. Here we do not further develop such aspects from the theory of free Lie algebras  \cite{reutenauer1993}, and refer to the aforementioned article for details on the subject.
     
\medskip         

The paper is organized as follows. The next section briefly recalls some general properties of graded Hopf algebras including the BWH decomposition of regularized Feynman rules viewed as Hopf algebra characters. We also dwell on the Fa\`a di Bruno Hopf algebra and prove an elementary but useful Lemma that allows the translation of the Dyson formula (relating bare and renormalized Green's functions) into the language of combinatorial Hopf algebras. In Section \ref{sect:Exp} we introduce the notion of $n$-regular characters and present an exponential recursion used to construct $m$-regular characters from $m-1$-regular ones. We conclude the article by introducing and studying a toy-model non-Rota--Baxter renormalization scheme on which the exponential recursion can be performed. We prove in particular that locality properties are preserved by this renormalization process.

%%%%%%%%%%%%%%%%%%%%%%%%%%%%%%%%%%%%%%%%%%%%%%%%%%%

\section{From Dyson to Fa\`a di Bruno}
\label{sect:D2F}

\subsection{Preliminaries}
\label{ssect:prelim}

In this section we introduce some mathematical structures to be used in the sequel. We also recall the BWH decomposition of Hopf algebra characters. Complementary details can be found, e.g. in \cite{EFGP, FGB, Manchon}.

Let us consider a graded, connected and commutative Hopf algebra $H=\bigoplus_{n \geq 0} H_n$ over the field $k$, or its pro-unipotent completion $\prod_{n \geq 0} H_n$. Recall that since the pioneering work of Pierre Cartier on formal groups \cite{CartierHGF}, it is well-known that the two types of Hopf algebras behave identically, allowing to deal similarly with finite sums $\sum_{n \leq N}h_n$, $h_n \in H_n$ and formal series $\sum_{n \in \NN}h_n$, $h_n \in H_n$. The unit in $H$ is denoted by $\un$. Natural candidates are the Hopf algebras of rooted trees and Feynman graphs~\cite{CKI,CKII} related to non-commutative geometry and pQFT, respectively. 

We remark here that graduation phenomena are essential for all our forthcoming computations, since in the examples of physical interest they incorporate information such as the number of loops (or vertices) in Feynman graphs, relevant for perturbative renormalization. The action of the grading operator $Y: H \to H$ is given by:
$$
    Y(h)  = \sum\limits_{n\in \mathbb{N}} n h_n \quad {\rm{for}} \quad h 
             = \sum\limits_{n\in \mathbb{N}}h_n\in\prod\limits_{n\in\mathbb{N}} H_n.
$$
We write $\epsilon$ for the augmentation map from $H$ to $H_0 = k \subset H$ and $H^+:=\bigoplus_{n=1}^{\infty} H_n$ for the augmentation ideal of $H$. The identity map of $H$ is denoted $id$. The product in $H$ is written $m_H$ and its action on elements simply by concatenation. The coproduct is written $\Delta$; we use Sweedler's notation and write $h^{(1)}\otimes h^{(2)}$ or $\sum_{j = 0}^n h_j^{(1)}\otimes h_{n-j}^{(2)}$ for $\Delta (h) \in \bigoplus_{j=0}^{n}H_j\otimes H_{n-j}$, $h\in H_n$. 

The space of $k$-linear maps from $H$ to $k$, $\Lin(H,k):=\prod_{n\in\mathbb{N}} \Lin(H_n,k)$, is naturally endowed with an associative unital algebra structure by the convolution product:
 \allowdisplaybreaks{
\begin{equation*}
    f\ast g := m_k\circ(f\otimes g)\circ\Delta : \qquad H
    \xrightarrow{\Delta} H \otimes H \xrightarrow{f \otimes g} k
    \otimes k \xrightarrow{m_{k}} k.
\end{equation*}}
The unit for the convolution product is precisely~the counit~$\epsilon : H \to k$. Recall that a character is a linear map $\gamma$ of unital algebras from $H$ to the base field $k$:
$$
    \gamma (hh') = \gamma(h)\gamma(h').
$$
The group of characters is denoted by $G$. With $\pi_n$, $n \in \mathbb{N}$, denoting the projection from  $H$ to $H_n$ we write $\gamma_{(n)}=\gamma \circ \pi_n$. An infinitesimal character is a linear map $\alpha$ from $H$ to $k$ such that:
$$
    \alpha (h h') = \alpha(h) \epsilon(h') + \epsilon (h) \alpha (h').
$$
As for characters, we write $\alpha(h) = \sum_{n \in\mathbb{N}} \alpha_{(n)} (h_n)$. We remark that by the definitions of characters and infinitesimal characters $\gamma_0(\un)=1$, that is $\gamma_0 = \epsilon$, whereas $\alpha_0(\un) = 0$, respectively. Recall that the graded  vector space $\frak{g}$ of infinitesimal characters is a Lie subalgebra of $\Lin(H,k)$ for the Lie bracket induced on the latter by the convolution product. 

Let $A$ be a commutative $k$-algebra, with unit~$1_A=\eta_A(1)$, $\eta_A : k \to A$ and with product~$m_A$, which we sometimes denote by a dot, i.e. $m_A(u\otimes v)=:u\cdot v$ or simply by concatenation. The main examples we have in mind are $A=\CC,\,A=\CC[[\varepsilon,\varepsilon^{-1}]$ and~$A=H$. We extend now the definition of characters and call an ($A$-valued) character of $H$ any algebra map from~$H$ to~$A$. In particular $H$-valued characters are simply algebra endomorphisms of $H$. We extend as well the notion of infinitesimal characters to maps from~$H$ to the commutative $k$-algebra $A$, that is:
$$
    \alpha(hh') = \alpha(h) \cdot e(h') + e(h) \cdot \alpha (h'),
$$
where $e:=\eta_A\circ\epsilon$ is now the unit in the convolution algebra $\Lin(H,A)$. Observe that infinitesimal characters can be alternatively defined as $k$-linear maps from $H$ to $A$ with $\alpha \circ \pi_0=0$ that vanish on the square of the augmentation ideal of $H$. The group (Lie algebra) of $A$-valued characters (infinitesimal characters) is denoted $G(A)$ ($\frak{g}(A)$) or $G_H(A)$ when we want to emphasize the underlying Hopf algebra.

%%%%%%%%%%%%%%%%%%%%%%%%%%%%%%%%%%%%%%%%%%%%%%%%%%

\subsection{Birkhoff--Wiener--Hopf decomposition of $G(A)$}
\label{ssect:BWH}

In the introduction we already mentioned one of Connes--Kreimer's seminal insights into the algebro-combinatorial structure underlying the process of perturbative renormalization in QFT. In the context of DR+MS, they reformulated the BPHZ-method as a Birkhoff--Wiener--Hopf decomposition of regularized Feynman rules, where the latter are seen as an element in the group $G(A)$. Of pivotal role in this approach is a Rota--Baxter algebra structure on the target algebra $A=\CC[[\varepsilon,\varepsilon^{-1}]$.  

In general, let us assume that the commutative algebra~$A = A_+ \oplus A_-$ splits directly into the subalgebras $A_\pm=T_\pm(A)$ with $1_A \in A_+$, defined in terms of the projectors $T_-$ and $T_+:=id-T_-$. The pair $(A,T_-)$ is a special case of a (weight one) Rota--Baxter algebra \cite{EFGP2} since $T_-$, and similarly $T_+$, satisfies the (weight one RB) relation:
\begin{equation}
    T_-(x)\cdot T_-(y) + T_-(x\cdot y) = T_-\bigl(T_-(x)\cdot y +
    x\cdot T_-(y)\bigr), \qquad x,y \in A.
    \label{eq:RBR}
\end{equation}
One easily shows that $\Lin(H,A)$ with an idempotent operator $\mathcal{T}_-$ defined by $\mathcal{T}_-(f)=T_- \circ f$, for $f\in \Lin(H,A)$, is a (in general non-commutative) unital Rota--Baxter algebra (of weight one). 

The Rota--Baxter property~\eqref{eq:RBR} implies that $G(A)$ decomposes as a set as the product of two subgroups:
$$
	G(A) = G_-(A)\ast G_+(A), \quad{\rm{where}}\quad G_\pm(A) 
	         = \exp^*(\mathcal{T}_\pm(\frak{g}(A))).
$$

\begin{cor} \cite{CKII,EFGP} \label{cor:ck-Birkhoff}
For any $\gamma \in G(A)$ the unique characters $\gamma_+\in G_{+}(A)$ and $\gamma_-^{-1}\in G_{-}(A)$ in the
decomposition of $G(A) = G_-(A)\ast G_+(A)$ solve the equations:
\begin{equation}
    \gamma_{\pm} = e \pm \mathcal{T}_{\pm}(\gamma_{-} \ast (\gamma - e)).
\label{eq:BogoliubovFormulae}
\end{equation}
That is, we have Connes--Kreimer's Birkhoff--Wiener--Hopf decomposition:
\begin{equation}
    \gamma = \gamma_-^{-1} \ast \gamma_+.
\label{eq:BCHbirkhoff}
\end{equation}
\end{cor}

Note that this corollary is true if and only if the operator $T_-$ on $A$ is of Rota--Baxter type. That is, uniqueness of the  decomposition follows from the idempotence of the map $T_-$. In fact, in the sequel we will show that this result is a special case of a more general decomposition of characters.

%%%%%%%%%%%%%%%%%%%%%%%%%%%%%%%%%%%%%%%%%%%%%%%%%%%

\subsection{The Fa\`a di Bruno Hopf algebra and a key lemma}
\label{ssect:FdBlemma}

Another example of combinatorial Hopf algebra, i.e. a graded, connected, commutative bialgebra with a basis indexed by combinatorial objects, which we will see to be acutely important in the sequel, is the famous Fa\`a di Bruno Hopf algebra $F$, for details see e.g.~\cite{BFFK,FGB,JR}. 

Recall that for series, say of a real variable $x$:
$$
	f(x) = \sum_{n=0}^\infty a_n(f)\,x^{n+1}, 
	\quad 
	h(x) = \sum_{n=0}^\infty a_n(h)\,x^{n+1}, 
	\ {\rm{with}}\ a_0(f) = a_0(h) = 1,
$$
the composition is given by:
$$
	f\bigl(h(x)\bigr) = \sum_{n=0}^\infty a_n\bigl(f\circ h\bigr)\,x^{n+1}=\sum\limits_{n=0}^\infty a_n(f)(h(x))^{n+1}.
$$
It defines the group structure on: 
$$
	G_F:=\Bigl\{f(x) = \sum_{n=0}^\infty a_n(f)\,x^{n+1} \ | \ a_n(f) \in \mathbb{C}, a_0(f) = 1\Bigr\}.
$$ 
One may interpret the functions $a_n$ as a derivation evaluated at $x = 0$:
$$
	a_n(f) =\frac{1}{(n+1)!}\frac{d^{n+1}f}{dx^{n+1}}(0).
$$
The coefficients $a_n\bigl(f \circ h\bigr)$ are given by:
$$
	a_n(f \circ h) = \sum_{k=0}^n a_k(f) \sum_{l_0 + \cdots + l_k=n-k \atop l_i \ge 0, i=0,\ldots,k}
	a_{l_0}(h)\cdots a_{l_k}(h).
$$
For instance, with an obvious notation, the coefficient of~$x^4$ in the composed series is given by
$$
	f_3 + 3f_2h_1 + f_1(h_1^2 + 2h_2) + h_3.
$$
The action of these coefficient functions on the elements of the group $G_F$ implies a pairing:
$$ 
	\langle a_n ,  f \rangle := a_n(f).
$$
The group structure on $G_F$ allows to define the structure of a commutative Hopf algebra on the polynomial ring spanned by the $a_n$, denoted by $F$, with coproduct:
$$
	\Delta_F(a_n)  = \sum_{k=0}^n \sum_{l_0 + \cdots + l_k=n-k \atop l_i \ge 0, i=0,\ldots,k}
						a_{l_0} \cdots a_{l_k} \otimes a_k .
$$ 
Notice that, using the pairing, an element $f$ of $G_F$ can be viewed as the $\RR[x]$-valued character $\hat{f}$ on $F$ characterized by: $\hat{f}(a_n):=a_n(f)x^n.$  The composition of formal power series translates then into the convolution product of characters: $\widehat{f(h)}=\hat{h} * \hat{f}$.

Let us condense this into what we call the Fa\`a di Bruno formula, that is,  define $a:=\sum_{n \ge 0} a_n$. Then $a$ satisfies:
\begin{equation}
\label{FdBformula}
	\Delta_F(a)  = \sum_{n \ge 0} a^{n+1} \otimes a_n.
\end{equation}

Note that subindices indicate the graduation degree. We prove now a technical lemma, important in view of applications to perturbative renormalization. As we will see further below, it allows to translate the Dyson formulas for renormalized and bare 1PI Green's functions into the language of Hopf algebras.

\begin{lem} \label{lem:FaBlemma}
Let $H=\prod_{n \geq 0} H_n$ be a complete graded commutative Hopf algebra, which is an algebra of formal power series containing the free variables $f_1,\ldots,f_n,\ldots$. We assume that $f_i$ has degree $i$ and write $f = \un + \sum_{k>0} f_k$ (so that, in particular, $f$ is invertible). If $\Delta(f)=\sum_n f \alpha^n\otimes f_n$, where  $\alpha=\sum_{n \ge 0} \alpha_n$ and the $\alpha_n$, $n>0$, are algebraically independent as well as algebraically independent from the $f_i$, then $\alpha$ satisfies the Fa\`a di Bruno formula:
$$
	\Delta(\alpha)=\sum\limits_{n \ge 0} \alpha^{n+1}\otimes \alpha_n.
$$
\end{lem}

\begin{proof}
Indeed, let us make explicit the associativity of the coproduct, $(\Delta \otimes id ) \circ \Delta=(id \otimes\Delta )\circ \Delta$. First:
$$
	\sum\limits_{n \ge 0} \Delta(f\alpha^n) \otimes f_n 
						= \sum\limits_{n \ge 0} f\alpha^n \otimes \Delta(f_n)
						= \sum\limits_{n,p \leq n}f\alpha^n\otimes (f\alpha^{n-p})_p \otimes f_{n-p}.
$$
Now we look at the component of this identity that lies in the subspace $H \otimes H \otimes H_1$ and get:
$$
	\Delta(f\alpha)=\sum\limits_{n \ge 0}f\alpha^n \otimes (f\alpha)_{n-1},
$$
that is:
$$
	\sum\limits_{n \ge 0} f\alpha^n\alpha_{(1)}\otimes f_n\alpha_{(2)}
				=\sum\limits_{n,p<n}f\alpha^n\otimes f_p\alpha_{n-p-1}.
$$
Since $f$ is invertible:
$$
	\sum\limits_{n \ge 0}\alpha^n\alpha_{(1)}\otimes f_n\alpha_{(2)}
				=\sum\limits_{n,p<n}\alpha^n\otimes f_p\alpha_{n-p-1}.
$$
From the assumption of algebraic independence among the $\alpha_i$ and $f_j$, we get, looking at the component associated to $f_0=1$ on the right hand side of the above tensor product:
$$
	\Delta(\alpha)	=	\alpha_{(1)}\otimes \alpha_{(2)}
				=	\sum\limits_{n \ge 0}\alpha^{n+1}\otimes \alpha_{n}.
$$
\end{proof}

\begin{cor}
With the hypothesis of the Lemma, the map $\chi$ from $F$ to $H$, $a_n\longmapsto \alpha_n$ is a Hopf algebra map. In particular, if $f$ and $g$ are in $G_H(\RR)$, $f\circ\chi$ and $g\circ\chi$ belong to $G_F(\RR)$ and:
$$
	\sum\limits_{n \ge 0} g\ast f(\alpha_n)x^{n+1} 	= \sum\limits_{n  \ge 0} (g\circ\chi) \ast (f\circ\chi) (a_n)x^{n+1}
								 		= f\circ\chi(g\circ\chi),
$$
where in the last equality we used the identification of $G_F(\RR)$ with $x+x\RR[x]$ to view $f\circ\chi$ and $g\circ \chi$ as formal power series.
\end{cor}

In other terms, properties of $H$ can be translated into the language of formal power series and their compositions.

%%%%%%%%%%%%%%%%%%%%%%%%%%%%%%%%%%%%%%%%%%%%%%%%%%%%

\section{The exponential method}
\label{sect:Exp} 

Let $H=\bigoplus_{n \geq 0} H_n$ be an arbitrary graded connected commutative Hopf algebra and $A$ a commutative $k$-algebra with unit $1_A=\eta_A(1)$. Recall that $\pi_n$ stands for the projection on $H_n$ orthogonally to the other graded components of $H$. As before, the group of characters with image in $A$ is denoted by $G(A)$, with unit $e:=\eta_A \circ \epsilon$. We assume in this section that the target algebra $A$ contains a subalgebra $A_+$, and that there is a linear projection map $P_+$ from $A$ onto $A_+$. We write $P_-:=id - P_+$.

The purpose of the present section is to construct a map from $G(A)$ to $G(A_+)$. In the particular case of a multiplicative renormalizable perturbative QFT, where $H$ is a Hopf algebra of Feynman diagrams and $A$ the target algebra of regularized Feynman rules, this map should send the corresponding Feynman rule character $\psi \in G(A)$ to a renormalized, but still regularized, Feynman rule character $R$.

The particular claim of $A_+ \subset A$ being a subalgebra implies $G(A_+)$ being a subgroup. This reflects the natural assumption, motivated by physics, that the resulting -renormalized- character  $R \in G(A_+)$ maps products of graphs into $A_+$, i.e. $R(\Gamma_1\Gamma_2)=R(\Gamma_1)R(\Gamma_2) \in A_+$. Or, to say the same, products of finite and regularized amplitudes are still finite. In the case where the target algebra has the Rota--Baxter property, the map from $G(A)$ to $G(A_+)$ should be induced by the BWH decomposition of characters.

%%%%%%%%%%%%%%%%%%%%%%%%%%%%%%%%%%%%%%%%%%%%%%%%%%%

\subsection{An algorithm for constructing regular characters}
\label{ssect:algorithm}

We first introduce the notion of {$n$-regular} characters. Later we identify them with characters renormalized up to degree $n$.

\begin{defn} \label{def:regular}
A character $\varphi \in G(A)$ is said to be regular up to order $n$, or $n$-regular, if ${P_+} \circ \varphi_{(l)} = \varphi_{(l)}$ for all  $l \leq n$.  A character is called regular if it is $n$-regular for all $n$. 
\end{defn}

In the next proposition we outline an iterative method to construct a regular character in $G(A_+)$ starting with an arbitrary one in $G(A)$. The iteration proceeds in terms of the grading of $H$. 

\begin{prop}   \label{prop:Exp1}
Let $\varphi \in G(A)$ be regular up to order $n$. Define $\mu^{\varphi}_{n+1}$ to be the linear map which is zero on $H_i$ for $i \not= n+1$ and: 
$$
	\mu^{\varphi}_{n+1}:=P_- \circ \varphi \circ \pi_{n+1}=P_- \circ \varphi_{(n+1)}.
$$ 
Then
\begin{enumerate}

\item
$\mu^{\varphi}_{n+1}$ is an infinitesimal character.  

\item 
The convolution exponential $\Upsilon^-_{n+1} := \exp^*{(-\mu^{\varphi}_{n+1})}$ is therefore a character.  
 
\item
The product $\varphi_{n+1}^+:=\Upsilon^-_{n+1}  \ast \varphi$ is a regular character up to order $n+1$. 

\end{enumerate}
\end{prop}

Note that we use the same notation for the projectors $P_\pm$ on $A$ and the ones defined on $\Lin(H,A)$. 

\begin{remark} 
Let us emphasize two crucial points. First, we see the algebraic naturalness of the particular assumption on $A_+$ being a subalgebra. Indeed, it allows for a simple construction of infinitesimal characters from characters in terms of the projector $P_-$. Second, at each order in the presented process we stay strictly inside the group $G(A)$. This property does not hold for other recursive renomalization algorithms. For example, in the BPHZ case, the recursion takes place in the larger algebra $Lin(H,A)$, see e.g. \cite{EFMP}.
\end{remark}

\begin{proof} 
Let us start by showing that $\mu^{\varphi}_{n+1}$ is an infinitesimal character. That is, its value is zero on any non trivial product of elements in $H$. In fact, for $y=xz \in H_{n+1}$, $x,z\notin H_0$, 
$$
	\mu^{\varphi}_{n+1} (y)	=	P_-(\varphi (y)) = P_-(\varphi (x)\varphi(z))
					     	= 	P_-({P_+}\varphi (x) {P_+}\phi (y)),
$$ 
since $\varphi$ is $n$-regular by assumption. This implies that $\mu^{\varphi}_{n+1}(y)=P_-\circ P_+(P_+\varphi (x)P_+\varphi(z))=0$ as the image of $P_+$ is a subalgebra in $A$.  

The second assertion is true for any infinitesimal character, see e.g. \cite{EFGP}.  

The third one follows from the next observations: 
\begin{itemize} 

\item 
For degree reasons (since $\mu^{\varphi}_{n+1}= 0$ on $H_k, \ k\leq n$), $\varphi_{n+1}^+=\exp^*(-\mu^{\varphi}_{n+1}) \ast \varphi = \varphi $ on $H_k, \ k\leq n$, so that $\exp^*(-\mu^{\varphi}_{n+1})\ast \varphi$ is regular up to order $n$.  

\item 
In degree $n+1$:  let $y \in H_{n+1}$. With a Sweedler-type notation for the reduced coproduct $\Delta (y) - y \otimes 1-1\otimes y = y'_{(1)}\otimes y'_{(2)}$, we get:
\allowdisplaybreaks{
\begin{eqnarray}
	\exp^*(-\mu^{\varphi}_{n+1})\ast \varphi (y)&=&
	\exp^*(-\mu^{\varphi}_{n+1})(y)+\varphi(y) + \exp^*(-\mu^{\varphi}_{n+1})(y'_{(1)})\varphi(y'_{(2)}) \nonumber\\
	                                           &=&-\mu^{\varphi}_{n+1}(y)+\varphi (y) = P_+\varphi(y), \label{almostreg}
\end{eqnarray}}	
which follows from $\exp^*(-\mu^{\varphi}_{n+1})$ being zero on $H_i$, $1 \leq i \leq n$ and $\exp^*(-\mu^{\varphi}_{n+1})=-\mu^{\varphi}_{n+1}$ on $H_{n+1}$.  Hence, this implies immediately:  
\allowdisplaybreaks{
\begin{eqnarray*}
	P_+((\exp^*(-\mu^{\varphi}_{n+1})\ast \varphi )(y))&=& P_+(P_+\varphi(y))\\
									            &=& P_+\varphi(y) 
									               =    (\exp^*(-\mu^{\varphi}_{n+1})\ast \varphi)(y). 
\end{eqnarray*}}									
\end{itemize}
\end{proof}

Note the following particular fact.  When iterating the above construction of regular characters, say, by going from a $n-1$-regular character $\varphi_{n-1}^+$ to the $n$-regular character $\varphi_{n}^+$, the $n-1$-regular character is by construction {\it{almost regular}} at order $n$. By this we mean that $\varphi_{n}^+(H_n)$ is given by applying $P_+$ to  $\varphi_{n-1}^+(H_n)$, see (\ref{almostreg}). This amounts to a simple subtraction, i.e. for $y \in H_n$:
$$
	\varphi_{n}^+(y) = P_+(\varphi_{n-1}^+(y))
	                               = \varphi_{n-1}^+(y) - P_-(\varphi_{n-1}^+(y)).
$$      
Observe that by construction for $y \in H_{n}$:
\begin{equation}
\label{preparation}
	\varphi_{n-1}^+(y)=\varphi_{n-2}^+(y) - P_-(\varphi_{n-2}^+(y^{(1)}_{n-1}))\varphi_{n-2}^+(y^{(2)}_1),
\end{equation} 
where the reader should recall the notation $\Delta(y)=\sum_{i=0}^n y_i^{(1)}\otimes y_{n-i}^{(2)}$ making the grading explicit in the coproduct. Further below we will interpret these results in the context of perturbative renormalization of Feynman graphs: for example, when $\Gamma$ is a UV divergent 1PI diagram of loop order $n$, the order one graph $\Gamma^{(2)}_1$ on the right-hand side of the formula consists of the unique one loop primitive cograph. That is, $\Gamma^{(2)}_1$ follows from $\Gamma$ with all its 1PI UV divergent subgraphs reduced to points. In the literature this is denoted as $res(\Gamma)=\Gamma^{(2)}_1$.

The following propositions capture the basic construction of a character regular to all orders from an arbitrary character. We call it the exponential method. 

\begin{prop} (Exponential method) \label{prop:Exp2}
We consider the recursion: $\Upsilon_{0}^- := e$, $\varphi_{0}^+:=\varphi$, and:
$$
	\varphi^+_{n+1}:= \Upsilon^-_{n+1} * \varphi_{n}^+,
$$ 
where $ \Upsilon_{n+1}^- :=\exp^*(- P_- \circ \varphi_{n}^+ \circ \pi_{n+1})$. Then, we have that $\varphi^+:=\varphi^+_{\infty} := \lim\limits_\rightarrow\varphi^+_{n}$ is regular to all orders. Moreover, $\Upsilon_{\infty}^- \ast \varphi= \varphi^+$, where:
$$
	\Upsilon_{\infty}^-:=\lim\limits_\rightarrow \Upsilon(n)
$$
and: 
$$
	\Upsilon(n):= \Upsilon^-_{n}  \ast \cdots \ast \Upsilon_{1}^-.
$$  
\end{prop}

\begin{remark} 
In the light of the application of the exponential method to perturbative renormalization in QFT, we introduce some useful terminology. We call  $\Upsilon_{l}^- :=\exp^*(- P_- \circ \varphi_{l-1}^+ \circ \pi_{l})$ the counterfactor of order $l$ and the product $\Upsilon(n):= \Upsilon^-_{n}  \ast \cdots \ast \Upsilon_{1}^- = \Upsilon^-_{n}  \ast \Upsilon(n-1)$ the counterterm of order $n$.  
\end{remark}

%%%%%%%%%%%%%%%%%%%%%%%%%%%%%%%%%%%%%%%%%%%%%%%%%%%

\subsection{On the construction of bare coupling constants}

The following two propositions will be of interest in the sequel when we dwell on the physical interpretation of the exponential method.  Let $A$ be as in Proposition \ref{prop:Exp1}. We introduce  a formal parameter $g$ which commutes with all elements in $A$, which we extend to the filtered complete algebra $A[[g]]$ (think of $g$ as the renormalized -i.e. finite- coupling constant of a QFT). The character $\varphi \in G(A)$ is extended to $\tilde\varphi \in G(A[[g]])$ so as to map $f = 1 + \sum_{k>0} f_k \in H$ to:
$$
	\tilde\varphi(f)(g) = 1 + \sum_{k>0} \varphi(f_k)g^k \in A[[g]].
$$
Notice that we emphasize the functional dependency of $\tilde\varphi(f)$ on $g$ for reasons that will become clear in our forthcoming developments.

Recall Lemma \ref{lem:FaBlemma}. We assume that $\Delta(f)=\sum_{n \ge 0} f \alpha^n \otimes f_n$ where $\alpha=\sum_{n \ge 0} \alpha_n \in H$ satisfies the Fa\`a di Bruno formula:
$$
	\Delta(\alpha)=\sum\limits_{n \ge 0} \alpha^{n+1} \otimes \alpha_n.
$$
Let $\tilde\varphi^+_{n+1}:= \tilde\Upsilon^-_{n+1} * \tilde\varphi_{n}^+ =\tilde\Upsilon(n+1) * \tilde\varphi \in G(A[[g]])$ be the $n+1$-regular character constructed via the exponential method from $\tilde\varphi \in G(A[[g]])$. Now we define for each counterfactor $\tilde\Upsilon^-_{l}$, $l\ge 0$ a formal power series in $g$, which we call the order $l$ bare coupling constant:
$$
	g_{(l)}(g) :=  \tilde\Upsilon^-_{l}(g\alpha)=g+\sum_{n > 0} a^{(l)}_n g^{n+1}, 
$$
$a^{(l)}_n:=\Upsilon^-_{l}( \alpha_n)$. Observe that $\tilde\Upsilon^-_{0}(g\alpha)=g\varepsilon(\un)=g$ and by construction $a^{(l)}_n=0$ for $n<l$. 

\begin{prop} (exponential counterterm and composition) \label{prop:ExpFaa1}
With the aforementioned assumptions, we find that applying the order $n$ counterterm $\tilde\Upsilon(n)$ to the series $g\alpha \in H$ equals the $n$-fold composition of the bare coupling constants $g_{(1)}(g), \cdots, g_{(n)}(g)$:
\allowdisplaybreaks{
\begin{eqnarray*}
	\tilde\Upsilon(n)(g\alpha)(g)&=& \tilde\Upsilon^-_{n}  \ast \cdots \ast \tilde\Upsilon_{1}^-(g\alpha)(g)\\
						  &=&g_{(1)}\circ \cdots \circ g_{(n)}(g).
\end{eqnarray*}}
\end{prop}

\begin{proof} 
The proof follows by induction together with the Fa\`a di Bruno formula.
\allowdisplaybreaks{
\begin{eqnarray*}
	\tilde\Upsilon(2)(g\alpha)(g)&=& \tilde\Upsilon^-_{2} \ast \tilde\Upsilon_{1}^-(g\alpha)(g)\\
				&=&\sum\limits_{n \ge 0}  (\tilde\Upsilon^-_{2}(\alpha)(g))^{n+1}a^{(1)}_ng^{n+1}\\	
				&=&\sum\limits_{n \ge 0}  a^{(1)}_n (\tilde\Upsilon^-_{2}(g\alpha)(g))^{n+1}
				  =   g_{(1)}\circ g_{(2)}(g).
\end{eqnarray*}}
Similarly:
\allowdisplaybreaks{
\begin{eqnarray*}
	\tilde\Upsilon(m)(g\alpha)(g)
	&=& \tilde\Upsilon^-_{m} \ast  \cdots \ast \tilde\Upsilon_{2}^- \ast \tilde\Upsilon_{1}^-(g\alpha)(g)\\
	&=&\sum\limits_{n \ge 0}  (\tilde\Upsilon^-_{m} \ast  \cdots \ast \tilde\Upsilon_{2}^- (\alpha)(g))^{n+1} a^{(1)}_ng^{n+1}\\	
	&=&\sum\limits_{n \ge 0} a_n^{(1)} (g_{(2)}\circ \cdots \circ g_{(m)})(g))^{n+1} 
	   =   g_{(1)}\circ (g_{(2)}\circ \cdots \circ g_{(m)})(g).
\end{eqnarray*}}
\end{proof}

\begin{prop} (exponential method and composition) \label{prop:ExpFaa2}
With the assumption of the foregoing proposition we find that:
\allowdisplaybreaks{
\begin{eqnarray*}
	\tilde\varphi^+_n(f)(g) &=&  \tilde\Upsilon(n)(f)(g)\cdot \tilde\varphi(f)\circ g_{(1)}\circ \cdots \circ g_{(n)} (g)
\end{eqnarray*}}
\end{prop}

\begin{proof} 
The proof follows from the coproduct $\Delta(f)=\sum_n f \alpha^n\otimes f_n$ by a simple calculation.
\allowdisplaybreaks{
\begin{eqnarray*}
	\tilde\varphi^+_n(f)(g) &=& (\tilde\Upsilon(n) * \tilde\varphi)(f)(g)\\
		&=& \sum_{m\ge 0} \tilde\Upsilon(n)(f \alpha^m)(g)  \varphi(f_m)g^m\\
		&=&  \tilde\Upsilon(n)(f)(g) \sum_{m\ge 0}\varphi(f_m)(\tilde\Upsilon(n)(g\alpha)(g))^m
\end{eqnarray*}}
from which we derive the above formula using Proposition \ref{prop:ExpFaa1}.
\end{proof}

The reader may recognize in this formula a familiar structure. This identity is indeed an elaboration on the Dyson formula -we shall return on this point later on.

%%%%%%%%%%%%%%%%%%%%%%%%%%%%%%%%%%%%%%%%%%%%%%%%%%%%%%%%%%%%%

\subsection{The BWH-decomposition as a special case}
\label{ssect:BWHexp}

The decomposition $\Upsilon^-_{\infty}\ast\varphi= \varphi^+$ in Proposition~\ref{prop:Exp2} may be interpreted as a generalized BWH decomposition. Indeed, under the Rota--Baxter assumption, that is if $P_-$ is a proper idempotent Rota--Baxter map (i.e. if the image of $P_-$ is a subalgebra, denoted $A_-$), $G(A)=G(A_-)\ast G(A_+)$ and the decomposition of a character $\varphi$ into the convolution product of an element in $G(A_-)$ and in $G(A_+)$ is necessarily unique (see \cite{EFGP2,EFMP} to which we refer for details on the Bogoliubov recursion in the context of Rota--Baxter algebras). In particular, $\Upsilon_\infty^-$ identifies with the counterterm $\varphi_-$ of the BWH decomposition.

Let us detail briefly this link with the BPHZ method under the Rota--Baxter assumption for the projection maps $P_-$ and $P_+$.

Proposition \ref{prop:Exp2} in the foregoing subsection leads to the following important remark (that holds independently of the RB assumption). Observe that by construction it is clear that for $y \in H_{k}$, $k<n+1$:
$$
	\Upsilon(n+1)(y)=\Upsilon(k)(y).
$$ 
Using $\varphi^+_{k-1}=\Upsilon(k-1)*\varphi$ we see with $y \in H_{k}$ that:
\allowdisplaybreaks{
 \begin{eqnarray}
	\Upsilon(k)(y) &=& \Upsilon^-_{k}  \ast \cdots \ast \Upsilon_{1}^-(y) 	\nonumber\\
	&=& -P_-(\varphi^+_{k-1}(y)) + \Upsilon(k-1)(y) 					\nonumber\\
	&=& -P_-(\varphi(y)) - P_-( \Upsilon(k-1)(y)  )- P_-( \Upsilon(k-1)(y_{(1)}') \varphi(y_{(2)}')  ) + \Upsilon(k-1)(y)\nonumber\\
	&=& -P_-(\varphi(y) + \Upsilon(k-1)(y'_{(1)}) \varphi(y'_{(2)})  ) + P_+( \Upsilon(k-1)(y)  ) \nonumber\\
	&=& -P_-(\Upsilon(k-1)*(\varphi-e) (y))+ P_+( \Upsilon(k-1)(y)  ).\label{BWHrb}
\end{eqnarray}}	

Now, note that for all $n > 0$, the RB property implies that $\Upsilon(n)(y)$ is in $A_-$ for $y \in H^+$.  Hence, going to equation (\ref{BWHrb}) we see that $ P_+( \Upsilon(k-1)(y)  ) = 0$.

\begin{prop}\label{prop:ZassenBogo}
For $n>0$ the characters $\varphi^+_n$ and $\Upsilon(n)$ restricted to $H^n:=\bigoplus_{i=0}^nH_i$ solve Bogoliubov's renormalization recursion.   
\end{prop}

\begin{proof} 
Let $x \in H^n$. From our previous discussion:
\allowdisplaybreaks{
\begin{eqnarray*}
	e(x) - P_- \circ (\Upsilon(n) * (\varphi - e))(x) &=& \Upsilon(n)(x).
\end{eqnarray*}}
Similarly:
\allowdisplaybreaks{
\begin{eqnarray*} 
	e(x) + P_+ \circ (\Upsilon(n) * (\varphi - e))(x) &=& e(x) + P_+ \circ (\Upsilon(n) * \varphi - \Upsilon(n))(x)\\
									    &=& e(x) + P_+ \circ (\varphi^+_n - \Upsilon(n))(x)\\  
									    &=& \varphi^+_n(x).
\end{eqnarray*}}

When going to  the last line we used $P_+ \circ P_-= P_- \circ P_+= 0$ as well as the Rota--Baxter property of $P_-$ and $P_+$. This implies that, on $H^n$, $\varphi^+_n= e + P_+\circ(\Upsilon(n) * (\varphi - e))$ and $\Upsilon(n)=e - P_-\circ(\Upsilon(n) * (\varphi - e))$ which are Bogoliubov's renormalization equations for the counterterm and the renormalized character, respectively, see e.g. \cite{EFGP2,EFMP}.
\end{proof}

%%%%%%%%%%%%%%%%%%%%%%%%%%%%%%%%%%%%%%%%%%%%%%%%%%%

\subsection{On counterterms in the BWH decomposition}
\label{ssect:Bogo}

Recall briefly how these results translate in the language of renormalization in perturbative QFT. This section also introduces several notations that will be useful later on. The reader is referred to the textbooks \cite{Collins,IZ1980} and the articles \cite{CKII,CKIII} for more details.

As often in the literature, the massless $\phi^4$ Lagrangian $L=L(\partial_\mu \phi,\phi, g)$ in four space-time dimensions shall serve as a paradigm:
\begin{equation}
\label{Lphi4}
	L  := \frac{1}{2} \partial_\mu \phi \partial^\mu \phi - \frac{g}{4!}\phi^4.
\end{equation} 
This is certainly a too simple Lagrangian to account for all the combinatorial subtleties of perturbative QFT, but its basic properties are quite enough for our present purpose. The quadratic part is called the free Lagrangian, denoted by $L_0$. The rest is called the interaction part, and is denoted by $L_i$. The parameter $g$ appearing in $L=L_0+L_i$ is the so-called renormalized, that is, finite coupling constant. 

Perturbation theory is most effectively expressed using Feynman graphs. Recall that from the above Lagrangian we can derive Feynman rules. Then any Feynman graph~$\Gamma$ corresponds by these Feynman rules to a Feynman amplitude. By $|\Gamma|$ we denote the number of loops in the diagram. Recall that in any given theory exists a rigid relation between the numbers of loops and vertices, for each given $m$-point function. In  $\phi^4$ theory, for graphs associated to the $2$-point function the number of vertices equals the number of loops. For graphs associated to the $4$-point function the number of vertices is equal to the number of loops plus one. A Feynman amplitude consists of the Feynman integral, i.e. a multiple $d(=4)$-dimensional momentum space integral:
\begin{equation}
	\Gamma \mapsto \bigg[\int\prod_{l=1}^{|\Gamma|}\,d^dk_l \bigg]I_\Gamma(p,k),
\label{eq:stone-of-contention}
\end{equation}
multiplied by a proper power of the coupling constant, i.e.   $g^{|\Gamma|+1}$ for $4$-point graphs and $g^{|\Gamma|}$ for $2$-point graphs. Here, $k=(k_1,\ldots,k_{|\Gamma|})$ are the $|\Gamma|$ independent internal (loop) momenta, that is, each independent loop yields one integration, and $p=(p_1,\ldots,p_N)$, with $\sum_{k=1}^N p_k=0$, denotes the~$N$ external momenta. Feynman integrals are most often divergent and require to be properly regularized and renormalized to acquire physical meaning. A regularization method is a prescription that parameterizes the divergencies appearing in Feynman amplitudes upon introducing non-physical parameters, denoted $\varepsilon$, thereby rendering them formally finite. Let us write $g\tilde\psi(\Gamma;\varepsilon)=g^{|\Gamma|+1}{\psi}(\Gamma;\varepsilon)$ for the regularized Feynman amplitude (for example in DR; the notation $\tilde\psi$ is introduced for later use).

Of pivotal interest are Green's functions, in particular 1PI $n$-point (regularized) Green's functions, denoted $G^{(n)}(g,\varepsilon):=G^{(n)}(p_1,\ldots,p_n;g,\varepsilon)$. In the following we will ignore the external momenta and omit the regularization parameter. Recall that for the renormalization of the Lagrangian (\ref{Lphi4}), the $4$- and $2$-legs 1PI Feynman graphs, respectively the corresponding amplitudes, beyond tree level are of particular interest.  As guiding examples we use therefore from now on the regularized momentum space 1PI $4$- and $2$-point Green's function. These are power series in the coupling $g$ with Feynman amplitudes as coefficients:
$$
	G^{(4)}(g)= \tilde\psi(gz_g) \quad {\rm{and}} \quad G^{(2)}(g)= \tilde\psi(z_\phi),
$$ 
where $z_g$ and $z_\phi$ stand for the formal coupling constant $z$-factors in the corresponding Hopf algebra of Feynman graphs $H$:
\begin{equation}
\label{z}
	z_g = \un +  \sum_{k>0} \Gamma^{(4)}_k  
	\quad {\rm{and}} \quad  
	z_\phi = \un - \sum_{k>1} \Gamma^{(2)}_k.
\end{equation}
Here, $\un$ is the empty graph in $H$ and: 
$$
	\Gamma^{(4)}_k := \sum_{m=1}^{N^{(4)}_k} \frac{\Gamma^{(4)}_{k,m}}
	                                                                           {sym(\Gamma^{(4)}_{k,m})} 
	\quad {\rm{and}} \quad 
	\Gamma^{(2)}_k :=\sum_{n=1}^{N^{(2)}_k} \frac{\Gamma^{(2)}_{k,n}}{sym(\Gamma^{(2)}_{k,n})}
$$ 
denote the sums of the $N^{(4)}_k$ 1PI 4-point and $N^{(2)}_k$ 2-point graphs of loop order $k$, divided by their symmetry factors, respectively. To deal with the polynomial dependency of the Green's functions on the coupling constant $g$, we write:
$$
	G^{(4)}(g)=g + \sum\limits_{k=1}^{\infty} g^{k+1}G^{(4)}_k
	\quad {\rm{and}} \quad  
	G^{(2)}(g)=1 - \sum\limits_{k=1}^{\infty} g^{k}G^{(2)}_k,
$$ 
so that $G^{(r)}_k=\psi (\Gamma^{(r)}_k)$, for $r=2,4$. 

Hence, as perturbative 1PI Green's functions are power series with individual -UV divergent- 1PI Feynman amplitudes as coefficients, one way to render them finite is to renormalize graph by graph. This is the purpose of the Bogoliubov recursion, which, in the context of DR+MS, was nicely encoded in the group-theoretical language by Connes and Kreimer \cite{CKII}. Indeed, let $H$ be the graded connected commutative Hopf algebra of 1PI Feynman graphs associated to the Lagrangian (\ref{Lphi4}) and let us choose the RB algebra of Laurent series $A=\CC[\varepsilon^{-1},\varepsilon]]$ as a target algebra for the regularized amplitudes (the natural choice in DR). Then, the correspondence $\Gamma\mapsto\tilde\psi(\Gamma;\varepsilon)$ extends uniquely to a character on $H$. That is, the regularized Feynman rules, $\tilde\psi$, can be interpreted as an element  of $G(A[[g]])$. 

Recall now that in the case of DR the underlying RB structure, i.e. the MS scheme, implies the unique BWH decomposition $\tilde\psi=\tilde\psi_-^{-1}\ast\tilde\psi_+$. This allows to recover Bogoliubov's classical counterterm map $C$ and the renormalized Feynman rules map $R$. Indeed, for an arbitrary 1PI graph $\Gamma \in H$, one gets:
$$
	C(\Gamma ) =\tilde\psi_-(\Gamma) 
		\quad {\rm{and}} \quad  
	 R(\Gamma )=\tilde\psi_+(\Gamma ).
$$
The linearity of $R$ then leads to renormalized 1PI Green's functions: $G_R^{(4)}(g)=R(gz_g)$, $G_R^{(2)}(g)=R(z_\phi)$. We refer to \cite{CKII} for further details.

%%%%%%%%%%%%%%%%%%%%%%%%%%%%%%%%%%%%%%%%%%%%%%%%%%%

\subsection{The Lagrangian picture}
\label{ssect:LagrangePic}

The counterterms $C(\Gamma )$ figure in the renormalization of the Lagrangian $L$. Indeed, for a multiplicative renormalizable QFT, it can be shown that the BPHZ method is equivalent to the method of additive, and hence multiplicative renormalization. Therefore, let us remind ourselves briefly of the additive method, characterized by adding order-by-order counterterms to the Lagrangian $L$. Eventually, this amounts to multiplying each term in the Lagrangian by particular renormalization factors. Details can be found in standard textbooks on perturbative QFT, such as \cite{Collins,IZ1980}. 

In general the additive renormalization prescription is defined as follows. The Lagrangian $L$ is modified by adding the so-called counterterm Lagrangian, $L_{ct}$, resulting in the renormalized Lagrangian:
$$
	L_{ren} :=L + L_{ct},
$$ 
where $L_{ct}:=\sum_{s>0} L_{ct}^{(s)}$ is defined by:
\begin{equation}
\label{renormalized0}
	L_{ct}:= C_1(g)\frac{1}{2} \partial_\mu \phi \partial^\mu \phi 
					- C_2(g)\frac{g}{4!}\phi^4,
\end{equation}
with $C_n(g):=\sum_{s>0}g^sC_n^{(s)}$, $n=1,2$ being power series in $g$. The $C_n^{(s)}$, $n=1,2$, $s>0$ are functions of the regularization parameter $\varepsilon$ to be defined iteratively as follows. 

To obtain the $1$-loop counterterm $L^{(1)}_{ct}$ one starts with $L=L_0+L_i$, computes the propagators and vertices, and generates all one-loop diagrams, that is, graphs of order $g^2$. Among those one isolates the $UV$ divergent 1PI Feynman diagrams and chooses the $1$-loop counterterm part $L_{ct}^{(1)}$, that is, $C_n^{(1)}$, $n=1,2$, so as to cancel these divergences. 

Now, use the $1$-loop renormalized Lagrangian $L_{ren} ^{(1)}:=L + L_{ct}^{(1)} + \sum_{s>1} L_{ct}^{(s)}$ to generate all graphs up to $2$-loops, that is, all graphs of order $g^3$. Note that this includes for instance graphs with one loop where one of the vertices is multiplied by $g^2C_2^{(1)}$ and the other one by $g$, leading to an order $g^3$ contribution. Again, as before, isolate the $UV$ divergent 1PI ones and choose the $2$-loop counterterm part $L_{ct}^{(2)}$, which is now of order $g^3$, again so as to cancel these divergencies. Proceed with the $2$-loop renormalized Lagrangian $L_{ren} ^{(2)}:=L + L_{ct}^{(1)}+L_{ct}^{(2)} + \sum_{s>2} L_{ct}^{(s)}$, and so on. The 2-point graphs contribute to the wave function counterterm, whereas 4-point graphs contribute to the coupling constant counterterm (see e.g. \cite[Chap. 5]{Collins}).

Note that after $j$ steps in the iterative prescription one obtains the resulting $j$th-loop renormalized Lagrangian: 
\begin{equation}
\label{renormalized1}
	L^{(j)}_{ren} :=  L_0+L_i + L^{(1)}_{ct} + \cdots + L^{(j)}_{ct}+ \sum_{s>j} L_{ct}^{(s)}
\end{equation}
with counterterms $C_n^{(s)}$, $n=1,2$ fixed  up to order $j$, such that it gives finite expressions up to loop order $j$. The part  $L_{ct}^{(s)}$, $s>j$, remains undetermined. In fact, later we will see that, in our terminology, some associated Feynman rules are $j$-regular.

The multiplicative renormalizability of $L$ implies that we may absorb the counterterms into the coupling constant and wave function $Z$-factors:
$$
	Z_g:=1+C_2(g),\ \ Z_\phi:=1+C_1(g),
$$
where $C_n(g)=\sum_{s=1} g^s C_n^{(s)}$, $n=1,2$.  We get:
\begin{equation}
\label{renormalized2a}
	L^{(j)}_{ren}  = \frac{1}{2}  Z_\phi \partial_\mu \phi \partial^\mu \phi - \frac{1}{4!}gZ_g\phi^4.
\end{equation}
As it turns out, Bogoliubov's counterterm map seen as $C \in G(A[[g]])$ gives:
$$
	Z_g(g)=C(z_g)(g) 	
	\quad {\rm{and}} \quad 
	Z_\phi(g)=C(z_\phi)(g),
$$
where we made the $g$ dependence explicit. 

Now we define the bare, or unrenormalized, field $\phi_{(0)}:= \sqrt{Z_{\phi}} \phi$ as well as the bare coupling constant:  
$$
	g^B(g):=\frac{gZ_g(g)}{{Z^2_{\phi}(g)}},
$$ 
and as $C \in G(A[[g]])$:
$$
	g^B(g)=g C(z_B)(g)
$$
where $ z_B:=z_g/z_{\phi}^2 \in H$ is the formal bare coupling. Up to the rescaling of the wave functions, the locality of the counterterms allows for the following renormalized Lagrangian:
\begin{equation}
\label{renormalized2b}
	L_{ren} =\frac{1}{2}  \partial_\mu \phi_{(0)} \partial^\mu \phi_{(0)} 
					- \frac{1}{4!}g^B(g)\phi^4_{(0)}.
\end{equation}

%%%%%%%%%%%%%%%%%%%%%%%%%%%%%%%%%%%%%%%%%%%%%%%%%%%%%%%%%%%%%

\subsection{Dyson's formula revisited}
\label{ssect:Dysonformula}

Let us denote once again by $H$ and $F$ the Hopf algebra of 1PI Feynman graphs of the massless $\phi^4$ theory in four space-time dimensions and the Fa\`a Di Bruno Hopf algebra, respectively. The purpose of the present section is to show how Dyson's formula, relating renormalized and (regularized) bare Green's functions, allows for a refined interpretation of the exponential method for constructing regular characters in the context of renormalization. 

We write $R$ and $C$  for the regularized renormalized Feynman rules and counterterm character, respectively.  
Recall the universal bare coupling constant: 
$$
	z_B:= z_g z_\phi^{-2}.
$$ 
It can be expanded as a formal series  in $H$:
\begin{equation}
\label{univcoup}
	z_B = \un + \sum_{k>0} \Gamma_k \in H,
\end{equation}
where $\Gamma_k \in H_k$ is a homogeneous polynomial of loop order $k$ in 1PI 2- and 4-point graphs with a linear part $\Gamma_k^{(4)} + 2 \Gamma_k^{(2)}$. Notice that, as $H$ is a polynomial algebra over Feynman graphs and since the family of the $\Gamma_k^{(4)}$ and of the $\Gamma_k^{(2)}$ are algebraically independent in $H$, also the families of $\Gamma_k$ and $\Gamma_k^{(r)}$, $r=2,4$, are algebraically independent in $H$.

Coming back to 1PI Green's functions. Dyson back then in the 1940s \cite{Dyson49} showed --in the context of QED, but the result holds in general \cite[Chap.8]{IZ1980}--  that the bare and renormalized 1PI $n$-point Green's functions satisfy the following simple identity:
\begin{equation}
\label{dyson}
	G_R^{(n)}(g) =Z^{n/2}_\phi G^{(n)}(g^B).
\end{equation}

Recall that the renormalized as well as the bare 2- and 4-point Green's functions and the $Z$-factors, $Z_\phi$ and $Z_g$, are obtained by applying respectively the renormalized Feynman rules map $R$,  the Feynman rules $\tilde\psi$ and the counterterm $C$ to the formal $z$-factors introduced in (\ref{z}), respectively. When translated into the language of Hopf algebras, the Dyson equation reads, say, in the case of the $4$-point function:
$$	
	G_R^{(4)}(g)=R(gz_g) = C(z^{2}_\phi) \sum_{j=0}^{\infty}C(z_B)^{j+1}\tilde\psi(g\Gamma_j^{(4)}) 
						       = \sum_{j=0}^{\infty}C(z_B)^{j}C(z_g) \tilde\psi(g\Gamma_j^{(4)}). 
$$
This can be rewritten:
\begin{eqnarray}
\label{E1}
	R(z_g)&=&m_A(C\otimes \tilde\psi)\sum_{j=0}^{\infty}z_B^{j}z_g\otimes\Gamma_j^{(4)}, 
\end{eqnarray}
where we recognize the convolution expression $R=C\ast \tilde\psi$ of the BWH decomposition, with:
\begin{eqnarray}
\label{E2}
	\Delta(z_g)=\sum_{k \ge 0} z_B^k z_g \otimes \Gamma^{(4)}_k.
\end{eqnarray}
Similarly, the study of the 2-point function yields:
$$
	\Delta(z_\phi)= z_\phi \otimes \un -  \sum_{k > 0} z_B^k z_\phi \otimes \Gamma^{(2)}_k.
$$

The equivalence between the two formulas (\ref{E1}) and (\ref{E2}) follow from the observation that the BWH decomposition of characters holds for arbitrary counterterms and renormalized characters, ${\tilde\psi}_-$ and ${\tilde\psi}_+$, respectively. Choosing, e.g. ${\tilde\psi}_-=C$ and ${\tilde\psi}_+=R$ in such a way that their values on Feynman diagrams form a family of algebraically independent elements (over the rationals) in $\CC$ shows that (\ref{E1}) implies (\ref{E2}) (the converse being obvious). Notice that the coproduct formulas can also be obtained directly from the combinatorics of Feynman graphs. We refer to \cite{Bel,CKIII,vanSu1,vanSu2} for complementary approaches and a self-contained study of coproduct formulas for the various formal $z$-factors. 

Now, Lemma \ref{lem:FaBlemma} implies immediately the Fa\`a di Bruno formula for $z_B$:

\begin{prop} \label{thething2}
$$	
	\Delta(z_B)=\sum_{k \ge 0} z_B^{k+1} \otimes \Gamma_k
$$
\end{prop}

\begin{cor}
There exists a natural Hopf algebra homomorphism $\Phi$ from $F$ to $H$: 
\begin{equation}
\label{FdB}
	a_n \mapsto \Phi(a_n) :=  \Gamma_n.
\end{equation}
\end{cor}

Equivalently, there exists a natural group homomorphism $\rho$ from $G(A)$, the $A$-valued character group of $H$, to the $A$-valued character group $G_{F}$ of $F$:
$$
	G(A) \ni \varphi \mapsto \rho(\varphi):=\varphi \circ \Phi : { F} \to A.
$$

%%%%%%%%%%%%%%%%%%%%%%%%%%%%%%%%%%%%%%%%%%%%%%%
 
\subsection{Dyson's formula and the exponential method}
\label{ssect:DysonExpo}

Let us briefly make explicit the exponential method for perturbative renormalization in the particular context of the Hopf algebras of renormalization. We denote by $H:=\bigoplus_{n\ge 0}H_n$ the Connes--Kreimer Hopf algebra of 1PI --UV-divergent-- Feynman graphs and by $G(A[[g]])$ the group of regularized characters from $H$ to the commutative unital algebra $A$ over $\mathbb{C}$ to be equipped with a $\mathbb{C}$-linear projector $P_-$ such that the image of $P_+:=id -P_-$ is a subalgebra. The algebra $A$ and projector $P_-$ reflect the regularization method respectively the renormalization scheme. The unit in $G(A[[g]])$ is denoted by $e$. The corresponding graded Lie algebra of infinitesimal characters is denoted by $\frak{g}(A[[g]])=\bigoplus_{n>0}\frak{g}_n(A[[g]])$. Let $\tilde\psi \in G(A[[g]])$ be the character corresponding to the regularized Feynman rules, derived from a Lagrangian of a -- multiplicative renormalizable -- perturbative quantum field theory, say, for instance $\phi^4$ in four space-time dimensions. Hence any $l$-loop graph $\Gamma \in H_l$  is mapped to: 
\begin{equation}
\label{FeynChar}
	\Gamma \xrightarrow{ \tilde\psi} \tilde\psi(\Gamma):=g^{|\Gamma|}\psi(\Gamma)=g^{l}\psi(\Gamma).
\end{equation}
Note that the character $\psi$ associates with a Feynman graph the corresponding Feynman integral whereas the character $\tilde\psi$ maps any graph with $|\Gamma|$ loops to its regularized Feynman integral multiplied by the $|\Gamma|$th power of the coupling constant. 

Recall that the exponential method of renormalization proceeds order-by-order in the number of loops. At one-loop order, one starts by considering the infinitesimal character of order one from $H$ to $A[[g]]$: 
$$
	 \tilde\tau_{1}:= P_- \circ \tilde\psi \circ \pi_1 \in \frak{g}_1(A[[g]]),
$$
The corresponding exponential {\it{counterfactor}} from $H$ to $A[[g]]$ is given by:
$$
	\tilde\Upsilon^-_{1}:=\exp^*(-  \tilde\tau_{1}).
$$
From the definition of the Feynman rules character (\ref{FeynChar}) we get: 
\allowdisplaybreaks{
\begin{eqnarray*}
	 \tilde\Upsilon^-_{1}(\Gamma_k)&=&\exp^*(-P_-\circ \tilde\psi\circ\pi_1)(\Gamma_k)\\
							  &=& g^{k}\exp^*(-P_-\circ \psi \circ\pi_1)(\Gamma_k)\\
							  &=& g^{k} \Upsilon^-_1(\Gamma_k).
\end{eqnarray*}}
The character:
$$
	 \tilde\psi^+_{1} :=  \tilde\Upsilon^-_{1} *  \tilde\psi
$$ 
is $1$-regular, i.e. it maps $H_1$ to $A_+[[g]]$. Indeed, as $h \in H_1$ is primitive we find $ \tilde\psi^+_{1}(h)= \tilde\psi(h) +  \tilde\Upsilon^-_{1}(h)= \tilde\psi(h)-P_-( \tilde\psi(h))=P_+( \tilde\psi(h))$. In general,  by multiplying the order $n-1$-regular character by the counterfactor $ \tilde\Upsilon^-_{n}$  we obtain the $n$-regular character:
$$
		 \tilde\psi^+_{n} :=  \tilde\Upsilon^-_{n} *  \tilde\psi^+_{n-1} =  \tilde\Upsilon(n) *  \tilde\psi,
$$
with the exponential order $n$ counterterm  $ \tilde\Upsilon(n):=  \tilde\Upsilon^-_{n} * \cdots *  \tilde\Upsilon^-_{1}$. Hence, in the Hopf algebra context the exponential method of iterative renormalization consists of a successive multiplicative construction of higher order regular characters from lower  order regular characters, obtained by multiplication with counterfactors. 

Next, we define the $n$th-order bare coupling constant:
$$
	g_n(g) =  \tilde\Upsilon^-_{n}(gz_B)(g) = g  + \sum_{k \ge 0} g^{k+1} \Upsilon^-_n(\Gamma_k) \in gA[[g]]. 
$$ 
Recall that $\Upsilon^-_n(\Gamma_k) = 0$ for $k<n$. We denote the $m$-fold iteration:
$$
	g_1 \circ \cdots \circ g_m (g) =: g^{ \circ}_{m}(g),
$$ 
where by Proposition~\ref{thething2} and from the general properties of Fa\`a di Bruno formulas, we have: $g_m^{\circ}(g)=\tilde\Upsilon(n)(gz_B)$.
We also introduce the  $n$th-order Z-factors:
$$
	Z^{(n)}_g(g):=   \tilde\Upsilon(n)(z_g)(g)
		\quad {\rm{and}} \quad 
	Z^{(n)}_\phi(g):=   \tilde\Upsilon(n)(z_\phi)(g),
$$
so that the $n$th-order renormalized 2- and 4-point 1PI Green's functions are:
\allowdisplaybreaks{
\begin{eqnarray*}
	 G_{R,n}^{(4)}(g):=g\tilde\psi^+_{n}(z_g)(g) &=&   g\tilde\Upsilon(n) *  \tilde\psi_g(z_g)(g) 
	   = \sum_{l \ge 0}   \tilde\Upsilon(n)(z_B^l z_g) g^{l+1}\psi(\Gamma^{(4)}_l) \\
	&=&    \tilde\Upsilon(n)( z_\phi)^2 \sum_{l \ge 0}   (\tilde\Upsilon(n)(gz_B)(g) )^{l+1}\psi(\Gamma^{(4)}_l)\\
	&=& (Z^{(n)}_\phi(g))^2 \sum_{l \ge 0}  {g_m^\circ(g)}^{l+1}\psi(\Gamma^{(4)}_l),
\end{eqnarray*}}
or, $G_{R,n}^{(4)}(g)=(Z_\phi^{(n)}(g))^2G^{(4)}(g_m^\circ(g))$. Similarly, $\tilde\psi^+_{n}(z_\phi)(g)= Z^{(n)}_\phi(g) \sum_{l \ge 0}  ( \tilde\Upsilon(n)(gz_B )(g))^{l}\psi(\Gamma^{(2)}_l)$ and $G_{R,n}^{(2)}(g)=Z_\phi^{(n)}(g)G^{(2)}(g_m^\circ(g))$. This corresponds to a Lagrangian multiplicatively renormalized up to order $n$:
$$
	L^{(n)}_{ren} :=\frac{1}{2} Z^{(n)}_\phi(g) \partial_\mu \phi \partial^\mu \phi  - \frac{gZ^{(n)}_g(g)}{4!}\phi^4.
$$
However, using Propositions \ref{prop:ExpFaa1} and \ref{prop:ExpFaa2}, we may also rescale the wave function and write:
$$
	L^{(n)}_{ren} :=\frac{1}{2}  \partial_\mu \phi_{n,0} \partial^\mu \phi_{n,0}  - \frac{g^{ \circ}_{n}(g)}{4!}\phi_{n,0}^4.
$$
where $\phi_{n,0}:=\sqrt{Z^{(n)}_\phi(g)} \phi$. Physically, on the level of the Lagrangian, the exponential renormalization method corresponds therefore to successive reparametrizations of the bare coupling constant.

%%%%%%%%%%%%%%%%%%%%%%%%%%%%%%%%%%%%%%%%%%%%%%%%
 
\section{On locality and non Rota--Baxter type subtraction schemes}
\label{sect:nonRB}

In this last section we present a class of non-Rota--Baxter type subtraction schemes combining the idea of fixing the values of Feynman rules at given values of the parameters and the minimal subtraction scheme in dimensional regularization. The latter is known to be local \cite{CasKen,Collins} and we will use this fact to prove that the new class of non-Rota--Baxter type schemes is local as well. 

We first introduce some terminology. Let $\psi$ denote a dimensionally regularized Feynman rules character corresponding to a perturbatively renormalizable (massless, for greater tractability) quantum field theory. It maps the graded connected Hopf algebra $H=\bigoplus_{n \ge 0}H_n$ of 1PI Feynman graphs into the algebra $A$ of Laurent series with finite pole part. In fact, to be more precise, the coefficients of such a Laurent series are functions of the external parameters. In this setting, Eq.~(\ref{eq:stone-of-contention}) specializes to (see e.g. \cite{Collins}):
$$
	H \ni \Gamma \mapsto \psi (\Gamma;\mu,g,s) = \sum_{n = -N}^{\infty} a^\mu_n(\Gamma;g,s) \varepsilon^{n}. 
$$
Here, $\mu$ denotes 'tHooft's mass, $\varepsilon$ the dimensional regularization parameter and $s$ the set of external parameters others than the coupling constant $g$. The algebra $A$ is equipped with a natural Rota--Baxter projector $T_-$ mapping any Laurent series to its pole part: 
$$
	T_-( \psi (\Gamma;\mu,g,s) ) :=  \sum_{n = -N}^{-1} a^\mu_n(\Gamma,g,s) \varepsilon^{n}. 
$$
This is equivalent to a direct decomposition of $A$ into the subalgebras $A_-:=T_-(A)$ and $A_+:=T_+(A)$.   

In this setting, recall that the BWH decomposition gives rise to a unique factorization:
$
	\psi=\psi^{-1}_-*\psi_+
$
into a counterterm map $\psi_-$ and the renormalized  Feynman rules map $\psi_+$. Both maps are characterised by Bogoliubov's renormalization recursions: 
$
	\psi_\pm = e \pm T_\pm \circ (\psi_- *(\psi - e)).  
$
The Rota--Baxter property of $T_-$ ensures that both, $\psi_-$ and $\psi_+$, are characters. 

Recall the notion of locality  \cite{CasKen,Collins}. We call a character $\psi$ (and, more generally, a linear form on $H$) strongly local if the coefficients in the Laurent series which it associates to graphs are polynomials in the external parameter. Notice that the convolution product of two strongly local characters is strongly local: strongly local characters form a subgroup of the group of characters. On the other hand a character $\psi$ is local if its counterterm $\psi_-$ is strongly local. Notice that strong locality implies locality. Indeed, since, by the Bogoliubov formula $\psi_-=e-T_-(\psi_-\circ(\psi-e))$, $\psi_-$ is strongly local if $\psi$ is strongly local due to the recursive nature of the formula.

It is well-known that for a multiplicatively renormalizable perturbative QFT with dimensionally regularized  Feynman rules character  $\psi$, the counterterm $\psi_-$ following from Bogoliubov's recursion is strongly local. 
Moreover, as the Birkhoff decomposition is unique, recall that comparing with the exponential method we get:    
$
	\psi_- =  \Upsilon_{\infty}^-:=\lim\limits_\rightarrow \Upsilon(n)
$
with:
$
	\Upsilon(n):= \Upsilon^-_{n}  \ast \cdots \ast \Upsilon_{1}^-.
$  
Hence, in the particular case of a Rota--Baxter type subtraction scheme the exponential method provides a decomposition of Bogoliubov's   counterterm character with respect to the grading of the Hopf algebra. The following Proposition shows that the exponential counterfactors inherit the strong locality property of the Bogoliubov's counterterm character.

\begin{prop}
In the context of minimal subtraction, the exponential counterfactors $\Upsilon^-_{i}$ and hence the exponential counterterms $\Upsilon(n)$ are strongly local iff $\psi_-= \Upsilon_{\infty}^-$ is strongly local .  
\end{prop}

\begin{proof}
One direction is evident as strong locality of the counterfactors implies strong locality of $ \Upsilon_{\infty}^-$. The proof of the opposite direction follows by induction. For any $\Gamma \in H_1$ we find:
$$
	\psi_-(\Gamma) = -T_-\circ \psi\circ\pi_1(\Gamma), 
$$ 
which implies that $-T_-\circ \psi\circ\pi_1$ is strongly local. The strong locality of $ \Upsilon_{1}^-:=\exp^*(-T_- \circ \psi \circ \pi_1)$ follows from the usual properties of the exponential map in a graded algebra. 

Let us assume that strong locality holds for $\Upsilon_{1}^-,\ldots,\Upsilon_n^-$. For $\Gamma \in H_{n+1}$ we find (for degree reasons):
$$
	\psi_-*\Upsilon^{-1}(n)(\Gamma) = \cdots * \Upsilon_{n+2}\ast\Upsilon_{n+1}^-(\Gamma)= \Upsilon_{n+1}^-(\Gamma) 
							  = -T_- \circ \psi^+_n \circ \pi_{n+1}(\Gamma).
$$
Strong locality of $-T_- \circ \psi^+_n \circ \pi_{n+1}$ follows, as well as strong locality of $\Upsilon_{n+1}^-=\exp^*(-T_- \circ \psi^+_n \circ \pi_{n+1})$.
\end{proof}

The next result with be useful later.

\begin{lem} \label{lem:help1}
For a strongly local character $\phi$ in the context of a proper projector $P_-$ on $A$, the exponential method leads to a decomposition $\phi= \Upsilon_{\infty}^- * \phi^+$ into a strongly local counterterm $ \Upsilon_{\infty}^-$ as well as a strongly local regular character $\phi^+$.  
\end{lem}

\begin{proof}
The proof follows once again from the definition of the recursion. Indeed, the first order counterfactor in the exponential method is:
$$
	\Upsilon_{1}^- = \exp^*(-T_- \circ \phi \circ \pi_1), 
$$ 
which is clearly strongly local, since $\phi$ is strongly local. Then $\phi^+_1= \Upsilon_{1}^- * \phi$ is strongly local as a product of strongly local characters. The same reasoning then applies at each order. 
\end{proof}

%%%%%%%%%%%%%%%%%%%%%%%%%%%%%%%%%%%%%%%%%%%%%%%%%%%%
 
\subsection{A non-Rota--Baxter subtraction scheme}
\label{ssect:nonRB}

We introduce now another projection, denoted $T^q_-$. It is a projector defined on $A$ in terms of the RB map $T_-$:
\begin{equation}
\label{Taylor}
	T_-^q:=T_- + \delta^n_{\varepsilon,q},
\end{equation}
where the linear map $\delta^n_{\varepsilon,q}$ is the  Taylor jet operator up to $n$th-order with respect to the variable $\varepsilon$ at zero, which evaluates the coefficient functions at all orders between $1$ and $n$ at the fixed value $q$: 
$$
	\delta_{\varepsilon,q}^n(\sum_{m=-N}^{\infty} a_m(x)\varepsilon^m):=\sum\limits_{i=1}^na_i(q)\varepsilon^i. 
$$
Note the condensed notation, where $q$ stands for a fixed set of values of parameters. The choice of the projection amounts, from the point of view of the renormalized quantities, to fix the coefficient functions at 0 for given values of  parameters (e.g. external momenta). One verifies that $T_-^q$ defines a linear projection. Moreover, the image of $T_+^q:=id - T_-^q$ forms a subalgebra in $A$ (the algebra of formal power series in $\varepsilon$ whose coefficient functions of order less than $n$ vanish at the chosen particular values $q$ of parameters), but the image of $T_-^q$ does not. This implies immediately that the projector $T_-^q$ is not of Rota--Baxter type.
Hence, we have in general:
$$
	T^q_-( \psi (\Gamma;\mu,g,s) ) =  \sum_{l = -N}^{-1} a^\mu_l(\Gamma,g,s) \varepsilon^{l} 
							+ \sum_{i=1}^na^\mu_i(\Gamma,g,q)\varepsilon^i
$$
and
$$
	T^q_+( \psi (\Gamma;\mu,g,s) ) =  \sum_{l = 0}^{\infty} a^\mu_l(\Gamma, g,s) \varepsilon^{l} 
							 - \sum_{i=1}^na^\mu_i(\Gamma,g,s)\varepsilon^i.
$$
We find:

\begin{prop}
Using the subtraction scheme defined in terms of projector $T_-^q$ on $A$, the exponential method applied to the Feynman rules character $\psi$ gives a regular character:
$$
	\psi_{q}^+=\Upsilon_{\infty,q}^- \ast \psi ,
$$ 
where we use a self-explaining notation for the counterterm $\Upsilon_{\infty,q}^- $ and the renormalized character $\psi_{q}^+$.
\end{prop}

Now we would like to prove that the exponential method using the projector $T_-^q$ on $A$ gives local counterterms. That is, we want to prove that the counterfactor $\Upsilon_{n,q}^-$ for all $n$, and hence $\Upsilon_{\infty,q}^-$, are strongly local.  In the following:
$$
	\psi_-= \Upsilon_{\infty}^- = \cdots  *\Upsilon_{n}^- *  \cdots  *\Upsilon_{2}^- *\Upsilon_{1}^-  
$$
stands for the multiplicative decomposition of Bogoliubov's strongly local counterterm character following from the exponential method using the minimal subtraction scheme $T_-$. Whereas:
 $$
	\Upsilon_{\infty,q}^- = \cdots  *\Upsilon_{n,q}^- *  \cdots * \Upsilon_{2,q}^- *  \Upsilon_{1,q}^-  
$$
stands for the counterterm character following from the exponential method according to the modified subtraction scheme $T_-^q$. The following Lemma is instrumental in this section.

\begin{lem}\label{lemtech}
For a substraction scheme such that the image of $P_+$ is a subalgebra, let $\phi$ be a $n$-regular character and $\xi$ be a regular character, then:
$$
	P_-\circ (\phi\ast \xi)_{n+1}=P_-\circ \phi_{n+1}.
$$
In particular, the counterfactor $\Upsilon^-_{n+1}$ associated to $\phi$ is equal to the counterfactor associated to $\phi\ast \xi$. 

It follows that, if the exponential decomposition of a character $\psi$ is given by: $\psi=\Upsilon{^-_\infty}\ast\psi^+$, the exponential decomposition of the convolution product of $\psi$ with a regular character $\xi$ is given by: $\psi\ast\xi=\Upsilon_\infty^-\ast(\psi^+\ast \xi)$.
\end{lem}

\begin{proof}
Indeed, for a $n+1$-loop graph $\Gamma$, $\phi\ast\xi(\Gamma)=\phi(\Gamma)+\xi(\Gamma) + c$, where $c$ is a linear combination of products of the image by $\phi$ and $\xi$ of graphs of loop-order strictly less than $n+1$. The regularity hypothesis and the hypothesis that the image of $P_+$ is a subalgebra imply $P_-(\xi(\Gamma)+c)=0$, hence the first assertion of the Lemma. The others follow from the definition of the exponential methods by recursion.
\end{proof}

\begin{lem}\label{truc} 
Let $\psi$ be a regular character for the minimal substraction scheme ($T_-\circ \psi=0$). Using the subtraction scheme defined in terms of projector $T_-^q$ on $A$, the exponential method applied to $\psi$ gives 
$
	\psi_{q}^+=\Upsilon_{\infty,q}^- \ast \psi ,
$ 
where, for each graph $\Gamma$, $\Upsilon_{\infty,q}^-(\Gamma)$ is a polynomial with constant coefficients in the perturbation parameter $\varepsilon$. In particular, $\Upsilon_{\infty,q}^-$ is strongly local.
\end{lem}

The Lemma follows from the definition of the substraction map $T_-^q$: by its very definition, since $\psi(\Gamma)$ is a formal power series in the parameter $\varepsilon$ (without singular part), $T_-^q\circ\psi (\Gamma)$ is  a polynomial (of degree less or equal to $n$) with constant coefficients in the perturbation parameter $\varepsilon$. As usual, this behaviour is preserved by convolution exponentials, and goes therefore recursively over to the $\Upsilon_{i,q}^-$ and to $\Upsilon_{\infty,q}^-$.

\begin{prop}  \label{prop:strongloc}
With the above hypothesis, i.e. a dimensionally regularized Feynman rules character $\psi$ which is local with respect to the minimal subtraction scheme, the counterfactors and counterterm of the exponential method, 
$
\Upsilon_{i,q}^-
$ 
respectively,
$ 
\Upsilon_{\infty,q},
$ 
obtained using the subtraction scheme defined in terms of the projector $T_-^q$ are strongly local. 
\end{prop}

\begin{proof}
Indeed, we have, using the MS scheme, the BWH decomposition $\psi=\psi_-^{-1}\ast\psi_+$, where $\psi_-^{-1}$ is strongly local.
Applying the exponential method with respect to the projector $T_-^q$ to $\psi^+$ we get, according to Lemma~\ref{truc}, a decomposition $\psi^+=\Upsilon_+^-\ast\psi_{++}$, where we write $\Upsilon_+^-$ (resp. $\psi_{++}$) for the counterterm and renormalized character and where $\Upsilon_+^-$ is strongly local. We get: $\psi=\psi_-^{-1}\ast\Upsilon_+^-\ast\psi_{++}$, where $\psi_{++}$ is regular with respect to $T_-^q$.

From Lemma~\ref{lemtech}, we know that the counterfactors and counterterm for $\psi$ in the exponential method for $T_-^q$ are equal to the counterfactors and counterterm for $\psi_-^{-1}\ast\Upsilon_+^-$, which is a product of strongly local characters, and therefore is strongly local. The Proposition follows then from Lemma~\ref{lem:help1} and its proof.
\end{proof}

\subsection{A Toy-model calculation}
\label{ssect:nonRBToy}

In the following example we apply the above introduced local non-Rota--Baxter type subtraction scheme within dimensional regularization. We exemplify it by means of a simple toy model calculation. We work with the bicommutative Hopf algebra $H^{lad} =\bigoplus_{k\ge 0} H^{lad}_k$ of rooted ladder trees. Let us recall the general coproduct of the tree $t_n$ with $n$ vertices:
$$
	\Delta(t_n)= t_n \otimes \un + \un \otimes t_n + \sum_{k=1}^{n-1} t_{n-k} \otimes t_k.
$$ 
The regularized toy model is defined by a character $\psi \in G(A)$ mapping the tree $t_n$ to an $n$-fold iterated Riemann integral with values in $A:=\CC[[\varepsilon,\varepsilon^{-1}]$:
\begin{equation}
\label{example}
	\psi(p;\varepsilon,\mu)(t_n) := \mu^{\varepsilon} \int_p^\infty \psi(x;\varepsilon,\mu)(t_{n-1}) \frac{dx}{x^{1+\varepsilon}} 
	                                                 = \frac{1}{n! \varepsilon^n} \exp\bigl(-n\varepsilon \log(\frac{p}{\mu})\bigr),
\end{equation} 
with $\psi(p;\varepsilon,\mu)(t_1):=\mu^{\varepsilon}\int_{p}^\infty \frac{dx}{x^{1+\varepsilon}}$, with $\mu,\varepsilon > 0$, and where $p$ denotes an external momenta. Recall that $\mu$ ('tHooft's mass) has been introduced for dimensional reasons, so as to make the ratio $\frac{p}{\mu}$ a dimensionless scalar. In the following we will write $a:=\log(\frac{p}{\mu})$ and $b:=\log(\frac{q}{\mu})$, where $q$ is fixed. For later use we write out the first three values:
\allowdisplaybreaks{
\begin{eqnarray*}
	\psi(p;\varepsilon,\mu)(t_1) &=& \frac{1}{\varepsilon} 
									- a 
									+ \frac{1}{2}\varepsilon a^2 
								   	- \frac{1}{3!}\varepsilon^2 a^3              
									+ \frac{1}{4!}\varepsilon^3 a^4 - O(\varepsilon^4) \\
	\psi(p;\varepsilon,\mu)(t_2) &=& \frac{1}{2\varepsilon^2} 
									-  \frac{1}{\varepsilon} a
									+ a^2 
								   	- \frac{2}{3}\varepsilon a^3              
									+ \frac{1}{3}\varepsilon^2 a^4
									- \frac{2}{15}\varepsilon^3 a^5 + O(\varepsilon^4) \\             
	\psi(p;\varepsilon,\mu)(t_3) &=& \frac{1}{3!\varepsilon^3} 
									-  \frac{1}{2\varepsilon^2} a 
									+ \frac{3}{4\varepsilon} a^2 
								   	- \frac{3}{4} a^3              
									+ \frac{9}{16}\varepsilon a^4
								          - \frac{27}{80}\varepsilon^2 a^5 
								         + O(\varepsilon^3).            
\end{eqnarray*}}

Now, for a Laurent series $\alpha(p/\mu):=\sum_{n=-N}^{\infty} \alpha_n(p/\mu)\varepsilon^n$, where the coefficients $\alpha_n=\alpha_n(p/\mu)$ are functions of $p/\mu$, we define the following projector $P_-$:
\begin{equation}
\label{nonRB1}
	P_-\bigl(\sum_{n=-N}^{\infty}\alpha_n(p/\mu)\varepsilon^n\bigr) := 
	\sum_{n=-N}^{-1}\alpha_n(p/\mu)\varepsilon^n + \alpha_1(q/\mu)\varepsilon,
\end{equation}
where $q$ is fixed and chosen appropriately. We get:
$$
	P_+\bigl(\sum_{n = -N}^\infty \alpha_n(p/\mu)\varepsilon^n\bigr) 
		= \alpha_0 + (\alpha_1(p/\mu) - \alpha_1(q/\mu))\varepsilon 
		+ \sum_{n = 2}^\infty \alpha_n(p/\mu)\varepsilon^n \in \CC[[\varepsilon]].
$$
One verifies that:
$$
	P_\pm^2=P_\pm \quad {\rm{and}}\quad P_\pm \circ P_\mp = P_\mp \circ P_\pm = 0.
$$ 

Let us emphasize that $P_-$ is not a Rota--Baxter map. This implies that we are not allowed to apply formulae (\ref{eq:BogoliubovFormulae}) in Corollary \ref{cor:ck-Birkhoff} for the renormalization of $\psi(p;\varepsilon,\mu)$.

However, we will show explicitly that the exponential method applies in this case, giving at each order a local counterterm(-factor) character as well as a finite renormalized character. At first order we apply the $1$-regular character, $\psi^+_1$, to the one vertex tree:
\allowdisplaybreaks{
\begin{eqnarray*}
	\psi^+_1(t_1)    =\Upsilon(1)\ast \psi(t_1)
			      &=& \bigl(\exp^*(-P_- \circ \psi \circ \pi_1) \ast \psi\bigr) (t_1)\\
			      &=& -P_- \circ \psi \circ \pi_1(t_1) + \psi (t_1)\\
			      &=& -P_-(\psi(t_1)) + \psi (t_1)\\
			      &=& -\bigl(\frac{1}{\varepsilon} + \frac{1}{2} \varepsilon b^2 \bigr) 
			      						+ \frac{1}{\varepsilon} 
									- a
									+ \frac{1}{2}\varepsilon a^2 
								   	- \frac{1}{3!}\varepsilon^2 a^3  + O(\varepsilon^3)   \\
			      &=&  - a + \frac{1}{2}\varepsilon (a^2 - b^2) + O(\varepsilon^2).
\end{eqnarray*}}
Observe that the counterfactor, and hence counterterm at order one is:
$$
	\Upsilon(1)(t_1) = \Upsilon^-_1(t_1) = \exp^*(-P_-\circ \psi \circ \pi_1) (t_1) 
				=  - \frac{1}{\varepsilon} - \frac{1}{2} \varepsilon b^2 
			 	= - \frac{1}{\varepsilon} (1 +  \frac{1}{2} \varepsilon^2 b^2),
$$
which is local, i.e. does not contain any $\log(p/\mu)$ terms. Let us define $f=f(\varepsilon;q):=1 +  \frac{1}{2}\varepsilon^2 b^2$. Now, calculate the $2$-regular character, $\psi^+_2$, on the two vertex tree:
\allowdisplaybreaks{
\begin{eqnarray*}
	\psi^+_2(t_2) = \Upsilon(2)\ast \psi(t_2) 
	&=& \bigl(\exp^*(-P_- \circ \psi^+_1 \circ \pi_2)  \ast \exp^*(-P_- \circ \psi \circ \pi_1) \ast \psi \bigr)(t_2)\\
	&=& \psi(t_2) + \Upsilon(1)(t_1)\psi (t_1) +  \Upsilon(2)(t_2) \\
	&=&  \psi(t_2) - P_- (\psi(t_1))\psi (t_1) -P_- (\psi^+_1(t_2)) + \frac{1}{2}P_- (\psi(t_1))P_- (\psi(t_1))\\
	&=&  \psi(t_2) - P_- (\psi(t_1))\psi (t_1) -P_-\bigl(\psi(t_2) - P_- (\psi(t_1))\psi (t_1) \bigr) \\
         &   & \quad\ -P_-\bigl(\frac{1}{2}P_- (\psi(t_1))P_- (\psi(t_1))\bigr)+ \frac{1}{2}P_- (\psi(t_1))P_- (\psi(t_1))\\
	&=& P_+\bigl( \psi(t_2) - P_- (\psi(t_1))\psi (t_1) \bigr) + \frac{1}{2}P_+\bigl(P_- (\psi(t_1))P_- (\psi(t_1))\bigr).
\end{eqnarray*}}
We first calculate the counterterm:
\allowdisplaybreaks{
\begin{eqnarray*}		
 \Upsilon(2)(t_2) &=& \exp^*(-P_- \circ \psi^+_1 \circ \pi_2) \ast \exp^*(-P_- \circ \psi \circ \pi_1)(t_2)\\ 
 			  &=& -P_- \circ \psi^+_1(t_2) +  \frac{1}{2}P_- (\psi(t_1))P_- (\psi(t_1))\\
			  &=&  -P_- ( \Upsilon^-_1*\psi(t_2)) 
			  		+ \frac{f^2}{2\varepsilon^2}
\end{eqnarray*}} 
Now observe that:
\allowdisplaybreaks{
\begin{eqnarray*}	
 \Upsilon^-_1*\psi(t_2) &=& \psi(t_2) - P_- (\psi(t_1))\psi (t_1) + \frac{1}{2}P_- (\psi(t_1))P_- (\psi(t_1))\\
 	&=& \frac{1}{2\varepsilon^2} 
									-  \frac{1}{\varepsilon} a
									+  a^2
								   	- \frac{2}{3}\varepsilon a^3              
									+ O(\varepsilon^2)   \\
					& & \qquad -\bigl(\frac{1}{\varepsilon} + \frac{1}{2} \varepsilon b^2 \bigr) 
			      		\bigl( \frac{1}{\varepsilon} 
									- a 
									+ \frac{1}{2}\varepsilon a^2
								   	- \frac{1}{3!}\varepsilon^2 a^3  + O(\varepsilon^3) \bigr) 
									+ \frac{1}{2}\bigl(\frac{1}{\varepsilon} + \frac{1}{2} \varepsilon b^2 \bigr) 
							      \bigl(\frac{1}{\varepsilon} + \frac{1}{2} \varepsilon b^2 \bigr)\\ 
		&=&  \frac{1}{2} a^2 - \frac{1}{2}\varepsilon\bigl(a^3 -  ab^2\bigr) + O(\varepsilon^2).
\end{eqnarray*}} 
We get:
$$
	-P_- ( \psi(t_2) - P_- (\psi(t_1))\psi (t_1) +  \Upsilon^-_1(t_2)) =  \Upsilon^-_2(t_2)=0.
$$
Hence, we find:
\allowdisplaybreaks{
\begin{eqnarray*}			 
		\Upsilon(2)(t_2) &=& \Upsilon^-_2\ast\Upsilon^-_1(t_2)=\Upsilon^-_1(t_2)
					   =  \frac{1}{2\varepsilon^2} f^2,
\end{eqnarray*}}
which is local, and:
\allowdisplaybreaks{
\begin{eqnarray*}
	\psi^+_2(t_2) =  \psi^+_1(t_2)
			   &=&  \frac{1}{2} a^2 - \frac{1}{2}\varepsilon\bigl(a^3 -  ab^2\bigr) + O(\varepsilon^2) \\
			   &=& \frac{1}{2}\bigl(  - a + \frac{1}{2}\varepsilon (a^2 - b^2)  + O(\varepsilon^2) \bigr)^2.
\end{eqnarray*}}

At third order, using $\Upsilon^-_2(t_2)= P_-(\psi^+_1(t_2))=0$ and $\Upsilon^-_2(t_1)=0$, a direct computation shows that similarly
the order 3 counterfactor, $\Upsilon^-_3$, evaluated on the order 3 tree, $t_3$, is zero:
$$
	\Upsilon^-_3(t_3)  = 	 \exp^*(-P_-\circ \psi^2_+ \circ \pi_3) (t_3) 
						= -P_-\bigl(\psi^2_+ (t_3) \bigr)
						=0
$$
whereas
\allowdisplaybreaks{
\begin{eqnarray*}
	\psi^+_3(t_3) &=& \psi^+_2(t_3) =\frac{1}{3!}\bigl(  - a + \frac{1}{2}\varepsilon (a^2 - b^2)  + O(\varepsilon^2) \bigr)^3,	
 \end{eqnarray*}}
 and the counterterm at order 3 is:
\allowdisplaybreaks{
\begin{eqnarray*}
	 \Upsilon(3)(t_3) &=&   -\frac{1}{3!\varepsilon^3}f^3.
 \end{eqnarray*}}
 This pattern is general and encoded
in the following proposition.
 
 \begin{prop} 
The renormalization of the toy-model (\ref{example}) via the exponential method, in the context  of DR together with the general non-RB scheme (\ref{Taylor}) gives the $n$th-order counterfactor $\Upsilon_n^-(t_n) = 0$ and counterterm:
$$
 	\Upsilon(n)(t_n) =\frac{1}{n!\varepsilon^n}(-f)^n,
$$
with: 
$$
	f=f(\varepsilon;q):=1 
									+ \frac{1}{2}\varepsilon^2 b^2 
								   	 - \frac{1}{3!}\varepsilon^3 b^3 
									+ \cdots 
									+ \frac{(-1)^{m+1}}{(m+1)!}\varepsilon^{m+1} b^{m+1}
$$
corresponding the Taylor jet operator (\ref{Taylor}), $\delta^m_{\varepsilon,q}$, say, of fixed order $m \in \mathbb{N}_+$. The $n$th-regular, i.e. renormalized character is given by:
$$
	\psi^+_n(t_n)= \frac{1}{n!}\bigl(  - a 
							+ \frac{1}{2}\varepsilon (a^2 - b^2) 
						   	- \frac{1}{3!}\varepsilon^2 (a^3 - b^3)
							+ \cdots               
							+ \frac{(-1)^m}{(m+1)!}\varepsilon^{m} (a^{m+1} - b^{m+1}) - O(\varepsilon^{m+1})\bigr)^n
$$
 \end{prop}

\begin{proof}
Let us write $T:=\un + \sum\limits_{n=1}^\infty t_n$ for the formal sum of all rooted ladder trees. This sum is a group-like element ($\Delta(T)=T\otimes T$). It follows that 
$$
	\psi(T)=\sum_n\frac{1}{n!}(\frac{1}{\varepsilon}(\frac{p}{\mu})^{-\varepsilon})^n
		  =\exp(\frac{1}{\varepsilon}(\frac{p}{\mu})^{-\varepsilon})
$$
can be rewritten as the convolution exponential of the infinitesimal character $\eta$:
$$
	\eta (t_n):= \begin{cases}
				\eta(t_1):=\frac{1}{\varepsilon}(\frac{p}{\mu})^{-\varepsilon}, & n = 1\\
				0	& else.
			\end{cases}
$$
Then: 
$$
	\psi(T) =\exp^\ast(\eta)(T).
$$
Let us write $\eta^-:=P_-(\eta)$ and $\eta^+:=P_+(\eta)$, so that, in particular $\eta^-(t_1)=-\Upsilon(1)(t_1)=\frac{f}{\varepsilon}$ and $\eta^+(t_1)= - a 
							+ \frac{1}{2}\varepsilon (a^2 - b^2) 
						   	- \frac{1}{3!}\varepsilon^2 (a^3 - b^3)
							+ \cdots               
							+ \frac{(-1)^m}{(m+1)!}\varepsilon^{m} (a^{m+1} - b^{m+1}) - O(\varepsilon^{m+1})$.

We get finally (recall that the convolution product of linear endomorphisms of a bicommutative Hopf algebra is commutative):
$$
	\psi=\exp^\ast(\eta)=\exp^\ast(\eta_-)\ast\exp^\ast(\eta_+),
$$
where $\Upsilon(1)^{-1}=\exp^\ast(\eta_-)$ and where (by direct inspection) $\exp^\ast(\eta_+)$ is regular. It follows that $\psi$ is renormalized already at the first order of the exponential algorithm, that is: $\Upsilon_\infty=\exp^\ast(-\eta_-)=\Upsilon(1)$ and $\psi^+=\exp^\ast(\eta_+)$. The Proposition follows from the group-like structure of $T$ which implies that:
$$
	\Upsilon_\infty(t_n)=\exp^\ast(-\eta_-)(t_n)=\frac{1}{n!}(-\eta_-(t_1))^n,
$$
and similarly for $\psi^+(t_n)$.
\end{proof}

Notice that in the classical MS scheme, one gets simply $f=1$ in the above formulas. One recovers then by the same arguments the well-known result following from the BPHZ method in DR and MS.	

\vspace{0.5cm}
\subsection*{Acknowledgments}

The first named author is supported by a de la Cierva grant from the Spanish government. We thank warmly J.~Gracia-Bond\'\i a. Long joint discussions on QFT in Nice and Zaragoza were seminal to the present work, which is part of a common long-term project.

%%%%%%%%%%%%%%%%%%%%%%%%%%%%%%%%%%%%%%%%%%%%%%%%%%%%%%%%%%%%%
\end{document}